\documentclass[aps,pra,amsmath,notitlepage,nofootinbib,twocolumn,superscriptaddress,noeprint]{revtex4-2}
\usepackage{amsmath,amsfonts,amssymb,amsthm}
\usepackage{mathtools}
\usepackage[pdftex]{graphicx}
\usepackage[usenames, dvipsnames, svgnames, table]{xcolor}
\usepackage[colorlinks=true, citecolor=Blue, linkcolor=BrickRed, urlcolor=Brown]{hyperref}
\usepackage{txfonts}
\usepackage{mathrsfs}
\usepackage{bm}
\usepackage{multirow}
\usepackage{braket}
\usepackage{import}
\usepackage[ruled, noline, noend]{algorithm2e}
\usepackage{array}
\usepackage{comment}
\usepackage{tikz}
\usetikzlibrary{quantikz}
\usepackage{graphicx}

\newtheorem{theorem}{Theorem}
\newtheorem{lemma}{Lemma}

\theoremstyle{definition}
\newtheorem{definition}{Definition}
\newcommand{\be}{\begin{equation}}
\newcommand{\ee}{\end{equation}}
\newcommand{\ben}{\begin{eqnarray}}
\newcommand{\een}{\end{eqnarray}}
\newcommand{\bes}{\begin{subequations}}
\newcommand{\ees}{\end{subequations}}
\newcommand{\bF}{\begin{figure}}
\newcommand{\eF}{\end{figure}}

\DeclareMathOperator{\tr}{Tr}

\newcommand{\orcid}[1]{\href{https://orcid.org/#1}{\includegraphics[height = 2ex]{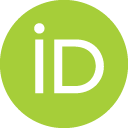}}}

\begin{document}

\title{Accreditation of Analogue Quantum Simulators}

\author{Andrew Jackson \orcid{0000-0002-5981-1604}}
\thanks{AJ and TK have contributed equally.}
\email{Andrew.Jackson.1@warwick.ac.uk}
\affiliation{Department of Physics, University of Warwick, Coventry CV4 7AL, United Kingdom}
  
\author{Theodoros Kapourniotis \orcid{0000-0002-6885-5916}}
\email{T.Kapourniotis@warwick.ac.uk }
\affiliation{Department of Physics, University of Warwick, Coventry CV4 7AL, United Kingdom}

\author{Animesh Datta \orcid{0000-0003-4021-4655}}
\email{animesh.datta@warwick.ac.uk}
\affiliation{Department of Physics, University of Warwick, Coventry CV4 7AL, United Kingdom}

%%%%%% Abstract %%%%%%

\begin{abstract}
   
   We present an accreditation protocol for analogue, i.e., continuous-time, quantum simulators.
   For a given simulation task, it provides an upper bound on the variation distance between the probability distributions at the output of an erroneous and error-free analogue quantum simulator.
   As its overheads are independent of the size and nature of the simulation, the protocol is ready for immediate usage and practical for the long term.
   It builds on the recent theoretical advances of strongly universal Hamiltonians and quantum accreditation as well as 
   experimental progress towards the realisation of programmable hybrid analogue–digital quantum simulators.

\end{abstract}

\date{\today}

\maketitle

%%%%%% MAIN TEXT %%%%%%
\section{Introduction}

Quantum simulation is rapidly emerging as a leading application of quantum technology~\cite{Altman2021}.
One key approach is analogue simulation,
which proceeds by engineering many-body quantum systems in a well-controlled environment and simply allowing their dynamics to occur.
As these systems increase in size and improve in performance, their computational capabilities are beginning to surpass those of existing classical computers~\cite{Bernien2017,Leseleuc2018}. 
Despite improvements, they continue to be afflicted by errors.
It is thus accepted that before analogue quantum simulators can tackle problems of  practical or fundamental importance, 
 methods to provide quantitative guarantees on the correctness of the outputs of error-prone analogue quantum simulators must be developed~\cite{Hauke_2012CanTrust} .

Validation of analogue quantum simulators has typically relied on tractable theoretical models incorporating errors and imperfections~\cite{Altman2021}.
Another method has been to run the dynamics forward and backward for equal amounts of time which returns the system to its initial state - should there be no errors. 
Commonly known as Loschmidt echo, this method can detect some errors and imperfections, 
but cannot provide quantitative guarantees on the correctness of the outputs.
More sophisticated variations have been developed that evolve the simulator
from some known initial state through a closed loop in state space, eventually returning to its initial state~\cite{Shaffer2021}.
These provide some measure of how faithfully the simulator implements the target Hamiltonian.
Methods such as randomised benchmarking have also been developed for analogue quantum simulators to quantify the performance of their components~\cite{2020AnalogueBenchmarking}.
However, these methods cannot provide quantitative guarantees on the correctness of the simulator's outputs either.

In this paper, we present a scalable and practical accreditation protocol that provides an upper bound on the correctness of the outputs of an analogue quantum simulator.
As the outputs of all quantum simulators are classical probability distributions, 
our protocol places an upper bound on the variation distance between the probability distributions generated by 
erroneous and error-free  analogue  quantum simulator.
We dub this task quantum accreditation.

Our protocol eliminates the need of classical simulations, thus freeing us to accredit simulations of arbitrarily large systems where quantum simulators offer the most. 
It is sensitive to a wide class of error processes affecting real-world analogue quantum simulators.
It can be implemented on extant programmable hybrid analogue-digital quantum simulators~\cite{Bluvstein2022}.
Our work can thus be construed to solve the open problem of verifiability for analogue quantum simulators by exploiting advances in their programmability~\cite[Sec. V]{Altman2021}.

The two obstacles to the quantum accreditation of analogue quantum simulators lie in the analogue and quantum natures of the problem. 
The former engenders a variety of Hamiltonians that such simulators can and do implement
~\cite{Nguyen2018,Brown2019, Nichols2019,Chiu2019,Browaeys2020,Christakis2023}, and 
starkly contrasts with the mathematical formulation of universal digital quantum computation~\cite{nielsen_chuang_2010}.
This has been a barrier to a general recipe to bound the correctness of outputs of analogue simulators.
The latter is the well-acknowledged exponential cost of simulating a general interacting many-body quantum system classically.

Our protocol overcomes the lack of universality plaguing analogue quantum simulators using the recently developed notion of universal 
quantum Hamiltonians~\cite{Cubitt9497} and their strong counterpart~\cite{zhou2021strongly}.
To overcome the latter obstacle, our protocol builds upon trap-based quantum interactive proof systems~\cite{PhysRevA.96.012303, 2018Theos, kapourniotis2022unifying}.
These have already been used to develop a scalable and practical accreditation protocol for the quantum circuit model~\cite{2019Sam,2021Sams}.
It implements the simulation of interest - the `target', together with a number of classically easy `trap' simulations of the same size as the target.
As they are implemented on the same hardware, they are subject to the same errors and the outputs of the traps imply an upper bound on the correctness of the target simulation.
The number of trap computations is independent of the size and nature of the target simulation, 
depending rather on the accuracy and confidence with which the upper bound is sought.

\section{Definitions}
\label{prelims}

We begin with a formal definition of analogue quantum simulators.

\begin{definition}
\label{AQSDef}

    An analogue quantum simulator takes as inputs:

    \begin{enumerate}
   \item The description of an initial product state $\ket{\psi_0}$,
     \item  A time-independent Hamiltonian, $\mathcal{H}_0$,
   \item A simulation duration, $t \in \mathbb{R}$, and
\item  A set of single-qubit measurements, $\mathcal{M}$.
\end{enumerate}

    The simulator prepares $\ket{\psi_0}$, applies the time evolution generated by $\mathcal{H}_0$ to $\ket{\psi_0}$ for the duration $t$, 
    followed by  the measurements in $\mathcal{M}$ and returns their results. These measurement outcomes will be samples from a distribution with probability measure, $P: \Omega \rightarrow [0,1]$, where $\Omega$ is the set of all possible outcomes of the measurements in $\mathcal{M}$.
\end{definition}

The accreditation of analogue quantum simulators relies on two recent advances -- one theoretical and one experimental.

The first, theoretical, advance is the notion of strongly universal Hamiltonians~\cite{zhou2021strongly}, which builds on the idea that the physics of any 
$\mathcal{O} \big( 1 \big)$-local quantum many-body system can be `simulated' by families of `universal' spin-lattice models~\cite{Cubitt9497}.  

\begin{definition}
A family of Hamiltonians is strongly universal if the eigenspectrum of any $O(1)$-local Hamiltonian can be encoded in some low-energy subspace of a Hamiltonian in the family; allowing the Hamiltonian from the family to simulate the $O(1)$-local Hamiltonian.
Moreover the translation from any $O(1)$-local Hamiltonian to one in the strongly universal family takes time at most polynomial in any relevant parameter
such as the number of qubits\footnote{Following Ref.~\cite{zhou2021strongly}, we use `qubits' instead of `sites'.} or interaction strength, 
and outputs a Hamiltonian  where any of these parameters are increased at most polynomially (in their original values).
\end{definition}

From all the possible families of strongly universal Hamiltonians, we choose that of XY-interactions on a square lattice, with freely varying coefficients~\cite{Cubitt9497}.
We focus on this Hamiltonian as its semi-translationally invariant nature~\cite{zhou2021strongly}
enables us to develop our accreditation protocol for a single form of interaction $(X_i X_j + Y_i Y_j)$, where $X_i, Y_i$ denote 
the Pauli X and Y operators respectively on qubit $i$.
Consequently, we dub the Hamiltonian of the XY-interaction on a square lattice `accreditable' in the rest of the paper.

\begin{definition}
\label{accreditableHamDef}

The family of accreditable Hamiltonians captures the XY-interactions on a square lattice, and has Hamiltonians of the form:
\begin{align}
    \label{accreditableHamiltonianDefEqn}
        \mathcal{H} =
        \sum_{\langle i,j \rangle} \bigg( J_{i,j} \bigg[ X_i X_j + Y_i Y_j \bigg] \bigg),
\end{align}
where $\langle i,j \rangle$ denotes the summation is over pairs of indices labelling qubits that are neighbours on the appropriately sized square lattice and $\forall i,j \in \mathbb{Z}, J_{i,j} \in \mathbb{R}$.
\end{definition}

Due to the strong universality of accreditable Hamiltonians, our protocol can be used to efficiently accredit any analogue quantum simulation
after translating it to the former using the constructive method in Ref. \cite{zhou2021strongly}.
This translation is efficient, but approximate with its precision captured by the spectral norm of the difference of the two Hamiltonians.
This is lack of precision is independent of the analogue quantum simulator on which the accreditable Hamiltonian is subsequently executed
and any errors afflicting its physical implementation.
The translation - encoding and decoding operations itself, when implemented on an analogue quantum simulator is subject to its errors.
This is accounted for in our protocol in Sec.~\ref{AccForStoc}.

The second, experimental, advance is the ability to apply single-qubit and two-qubit operations mid-simulation~\cite{Bluvstein2022}, which
 can be thought of as `single-qubit and two-qubit gates'. We formalise this in the definition of a hybrid quantum simulator.

\begin{definition}
\label{HQSDef}

    A hybrid quantum simulator (HQS) takes the four inputs of the analogue quantum simulator in Definition~\ref{AQSDef} and
    \begin{enumerate}
    \setcounter{enumi}{4}
    \item an ordered set, $G$, of single-qubit or two-qubit quantum gates with corresponding time-stamps $\{t_g \vert g \in G \}$ denoting when they are applied. 
    \end{enumerate}

The HQS prepares $\ket{\psi_0}$, applies the time evolution generated by $\mathcal{H}_0$ to $\ket{\psi_0}$ for the duration $t$, 
with interruptions to apply each gate $g$ in $G$ at time $t_g$, followed by  the measurements in $\mathcal{M}$ and returns their results. 
These measurement outcomes will be samples from a distribution with probability measure, $P: \Omega \rightarrow [0,1]$, where $\Omega$ is the set of all possible outcomes of the measurements in $\mathcal{M}$.
\end{definition}

Our accreditation protocol can be implemented on the state-of-the-art HQSs already extant.
Instances include quantum simulators where the XY-interaction Hamiltonians are directly implementable in experimental systems, most notably using Rydberg atoms~\cite{Browaeys2020} are embracing the ability to perform digital gates alongside analogue simulations~\cite{Bluvstein2022}.
Recently, hybrid models of digital-analogue quantum computation have also been studied theoretically to investigate their computational capabilities~\cite{Parra_Rodriguez_2020}.

    \begin{figure*}[t]
    \centering
 \begin{quantikz}[row sep=0.3cm]
% First qubit
\lstick{$\ket{0}$} & \gate{A_1} & \gate[wires=5, nwires=4]{\mathcal{U}_1} &\qw & \gate[wires=5, nwires=4]{e^{-i \mathcal{H} t / 2}} & \gate{B_1} & \gate[wires=5, nwires=4]{e^{-i \mathcal{H} t / 2}} & \gate{C_1} & \gate[wires=5, nwires=4]{\mathcal{U}_2} & \gate{D_1} & \meterD{Z}\\
% Second qubit
\lstick{$\ket{0}$} & \gate{A_2} & & \qw &  &  \gate{B_2} & & \gate{C_2} & &\gate{D_2} &  \meterD{Z}\\
% Third qubit
\lstick{$\ket{0}$} & \gate{A_3} && \qw & &  \gate{B_3} && \gate{C_3} & &\gate{D_3} &  \meterD{Z}\\
\vdots & \vdots &  \vdots & & \vdots & \vdots &  & \vdots &  \vdots & \vdots\\
% Fourth qubit
\lstick{$\ket{0}$} & \gate{A_N} && \qw &  &  \gate{B_N} & &\gate{C_N}&  & \gate{D_N} &  \meterD{Z}
\end{quantikz}
    \caption{Circuit representation of a HQS:
     $N$ is the number of qubits and $\mathcal{U}_1$, $\mathcal{U}_2$ are arbitrary $\mathrm{poly}(N)$-sized circuits of single and two qubit gates.
    Each $A_j$, $B_j$, $C_j$, $D_j$ (for any $j \in \mathbb{N}^{\leq N}$) is an arbitrary single-qubit gate. The input is fixed to  $\ket{\psi_0}=\ket{0}^{\otimes N}$ for convenience. 
    As the single-qubit gates $A_j$ are arbitrary, 
    any product state can be prepared as an input to the HQS. Similarly as the single-qubit gates $D_j$ are arbitrary,
     any product measurement can be performed on the output of the HQS.}
    \label{fig:generalFormForErrorfulSinglesRepeat1}
\end{figure*}
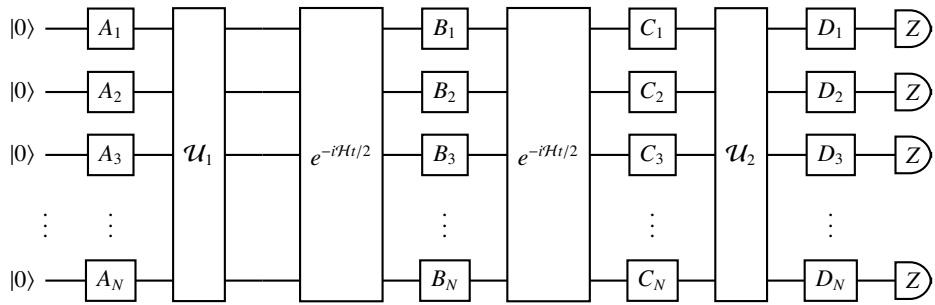

For the rest of the paper we depict the operations in a HQS, including both the Hamiltonian evolutions and the gates, via circuits, as illustrated in Fig.~\ref{fig:generalFormForErrorfulSinglesRepeat1}. 
It is important to note that our `circuits' do not necessarily fit within the limits of the quantum circuit model: The set of allowed operators in the HQS cannot be encapsulated by a finite gate-set as it contains time evolutions of any permissible Hamiltonian for arbitrary time.

\section{Accredited analogue quantum simulation}
\label{AccForStoc}

All physical implementations of quantum simulators will be afflicted by error. 
These errors can be considered to occur in one (or more) of the inputs listed in Definitions \ref{AQSDef} and \ref{HQSDef} being implemented incorrectly 
or affected by noise.
For instance, there could occur some fluctuations in the applications of the Hamiltonian $\mathcal{H}_0$
or an error in the value of $t$ applied.
Consequently, the outputs actually obtained will be erroneous. 
Thus the probability measure over $\Omega \ni s$, $\Tilde{P}(s)$, in the actual, erroneous case differs from the error-free probability measure, $P(s)$.

The objective of an accredited analogue quantum simulation is to provide -- 
in addition to the output $\Tilde{P}(s)$ -- an experimentally accessible upper-bound on the distance between $\Tilde{P}(s)$ and $P(s)$, 
as captured by the total variation distance, defined in Eqn.~\eqref{eq:tvd} in the following definition.

\begin{definition}
\label{AAQSDef}

An accredited analogue quantum simulation runs on a HQS. 
It takes the four inputs of the analogue quantum simulator in Definition~\ref{AQSDef} and two parameters $\alpha, \theta \in [0,1)$. It returns the outputs of the analogue quantum simulator, and an $\epsilon \in [0,1]$, such that:
\be
\label{eq:tvd}
      \mathrm{VD} \big( P, \Tilde{P} \big)
        =
        \dfrac{1}{2} \sum_{s \in \Omega} \bigg \vert P \big( s \big) - \Tilde{P}\big( s \big) \bigg \vert
        \leq
        \epsilon, 
       \ee
where the $\epsilon$ is obtained from the experimentally estimated $\Tilde{P}\big( s \big)$ with accuracy $\theta$ and confidence $\alpha$. 
\end{definition}

\subsection{Error model}
\label{ErrorModelSubsection}
We model any erroneous implementation of \emph{any part} of a HQS
as their error-free implementation followed (or preceded) by an error operator such that:
\begin{itemize}
    \item[E1.] The error operator is a completely positive trace preserving (CPTP) map applied on the HQS and its environment. 	
    \item[E2:] The error operator in distinct uses of a HQS, where each use only differs in the single-qubit gates, is independent and identically distributed (i.i.d). 	
    \item[E3:] Replacing identity gates with single-qubit gates in a simulation does not reduce the variational distance between $P \big( s \big)$ and $ \Tilde{P}\big( s \big)$. 
\end{itemize}

Our error model captures a large class of errors that afflict real-world analogue quantum simulators.
These include spontaneous emission, crosstalk, and particle loss via E1, and fast noise~\cite{Shaffer2021} such as laser fluctuations via E2. 
More generally, E2 captures all fluctuations in the HQS and its environment that occur on timescales faster than that implements the operations in Fig.~\ref{fig:generalFormForErrorfulSinglesRepeat1}.
Other common practical issues such as miscalibrations~\cite{Shaffer2021} in the duration of time evolutions or coefficients of the Hamiltonian being applied, or
unintended terms in the Hamiltonian, and incorrect state-preparation or measurement that occur repeatedly across multiple implementations of Fig.~\ref{fig:generalFormForErrorfulSinglesRepeat1} are captured as well, subject to E3.

Our error model does not capture slow noise~\cite{Shaffer2021} processes such as those from temperature variations or 
degradation of device performance over implementations of Fig.~\ref{fig:generalFormForErrorfulSinglesRepeat1}.
This is because they violate E2. We discuss means of mitigating this limitation later.

This error model affecting a HQS admits the following mathematical simplification.

\begin{lemma}
\label{simplerErrorLem}    
Any HQS affected by errors obeying E1-E3 is equivalent to the single-qubit gates being error-free and the remaining error being independent of $A_j, B_j, D_j$ ($\forall j \in \mathbb{N}^{\leq N}$). 
\end{lemma}
Notably the Lemma~\ref{simplerErrorLem} excludes $C_j$ from the single-qubit gates that the remaining error is independent of. This follows from the $C_j$ gates' unfortunate position in the simulation, as can be seen in the proof of Lemma~\ref{simplerErrorLem}, in Appendix \ref{SimplifyingErrorApp}. 
However, this dependence in no way interferes with our quantum accreditation protocol presented below.

\subsection{Quantum Accreditation Protocol}
\label{protocolSubsection}
We now sketch of our accreditation protocol, as presented in Protocol \ref{fullProtocolSketch}.
Ours is a trap-based protocol, inspired by quantum interactive proof systems~\cite{2018Theos} and 
recently adapted for the accreditation of digital quantum computation in the circuit model~\cite{2021Sams}.
As such, it is based on two different types of simulations -- the target and the trap.
The target simulation is the one of interest while the trap simulations are factitious ones
that exist to infer the effect of error in the target simulations that implemented on a HQS obeying E1-E3 .

\emph{Target simulations}:
This is the analogue quantum simulation we are actually interested in.
For it to be accredited, it is implemented on a HQS as illustrated in Fig.~\ref{fig:targInSplitUpWITHV}. 
In the absence of any errors, it applies $e^{-i \mathcal{H} t}$ to $\vert \psi_0 \rangle$ followed by single-qubit measurements. 
It requires that the initial states $\ket{\psi_0}= \otimes_{j = 1}^N \big( A^{\prime}_j \big) \ket{0}^{\otimes n}$ be encoded to enable simulation by the strongly universal Hamiltonian. This is denoted by $\mathcal{V}$.
Similarly, after the time evolution the state must be decoded, before measurement by $\mathcal{M}= \otimes_{j = 1}^N \big( D^{\prime}_j \big) Z^{\otimes n}$.
 This is done via $\mathcal{V}^{-1}$.

\emph{Trap simulations}:
\label{basicTrapStochSubsection}
 For any target simulation as in  Fig.~\ref{fig:targInSplitUpWITHV}, a trap simulation can be obtained by replacing the identities ($I$) with single-qubit gates 
 that invert the time evolution of the Hamiltonian 
 and changing some pre-existing single-qubit gates depending on some random parameters, as explained in the caption of Fig.~\ref{fig:PauliTwirledTrapCircuit(withHadamardsandConglom)WITHV}.
 The former is detailed in Sec.~\ref{sec:trapdesign}.
In the absence of any errors, the trap simulation executes the identity evolution, which will result in the all-zero output.
This can be checked using resources scaling linearly with the problem size $N$.
Any deviation from an all-zero output indicates the presence of errors.

Both target and trap simulations (Figs.~\ref{fig:targInSplitUpWITHV} and \ref{fig:PauliTwirledTrapCircuit(withHadamardsandConglom)WITHV} respectively)
can be implemented on a HQS (as in Fig.~~\ref{fig:generalFormForErrorfulSinglesRepeat1}) with 
$A_i = A_i', B_i = C_i= I, D_i=D_i', \mathcal{U}_1=\mathcal{V}, \mathcal{U}_2 = \mathcal{V}^{-1},$
and
$ A_i = P_iH^hZ', B_i = C_i = \mathcal{C}_i, D_i=Z'H^hP_i, \mathcal{U}_1=\mathcal{V} =\mathcal{U}_2 = \mathcal{V}^{-1}$
respectively.

The goal of the trap simulations is to detect any error on the HQS that obeys E1-E3 and provide a bound of the form in Eqn.~\ref{eq:tvd}.
Lemma~\ref{ProbCompare} below establishes a relationship between the effects of these errors in the trap simulations and the target simulation.
Thus, detecting the errors in the former, via Lemma~\ref{stochDetectLemma} enables us to bound the
variational distance between the error-free and erroneous probability distributions over the measurement outcomes of the latter, as per Theorem~\ref{finalTheorem}.

\begin{figure*}
    \centering
 \begin{quantikz}[row sep=0.3cm]
% First qubit
\lstick{$\ket{0}$} & \gate{A^{\prime}_1} & \gate[wires=5, nwires=4]{\mathcal{V}} &\gate[wires=5, nwires=4]{e^{-i \mathcal{H} t / 2}} & \gate{I} & \gate[wires=5, nwires=4]{e^{-i \mathcal{H} t / 2}} & \gate{I} &
 \gate[wires=5, nwires=4]{\mathcal{V}^{-1}} & \gate{D^{\prime}_1} & \meterD{Z}\\
% Second qubit
\lstick{$\ket{0}$} & \gate{A^{\prime}_2} &  &  & \gate{I} & & \gate{I} & &\gate{D^{\prime}_2}&  \meterD{Z}\\
% Third qubit
\lstick{$\ket{0}$} & \gate{A^{\prime}_3} & &  & \gate{I} & & \gate{I} & &\gate{D^{\prime}_3} & \meterD{Z}\\
\vdots & \vdots &  \vdots & & \vdots & &\vdots&& \vdots & \vdots \\
% Fourth qubit
\lstick{$\ket{0}$} & \gate{A^{\prime}_N}& &  & \gate{I} & & \gate{I} & & \gate{D^{\prime}_N} & \meterD{Z}
\end{quantikz}
    \caption{Target simulation with encoding of initial state $\ket{\psi_0}= \otimes_{j = 1}^N \big( A^{\prime}_j \big) \ket{0}^{\otimes n}$ by $\mathcal{V}$, 
    and decoding after the time evolution by $\mathcal{V}^{-1}$ before the measurements $\mathcal{M}= \otimes_{j = 1}^N \big( D^{\prime}_j \big) Z^{\otimes N}$. 
    Some of the $N$ qubits may be used as ancilla for preparing the encoded states via $\mathcal{V}$.
    The identity operations ($I$) are long enough in duration to allow for the implementation of single-qubit gates, as explained in Sec.~\ref{sec:trapdesign}.}
    \label{fig:targInSplitUpWITHV}
\end{figure*}
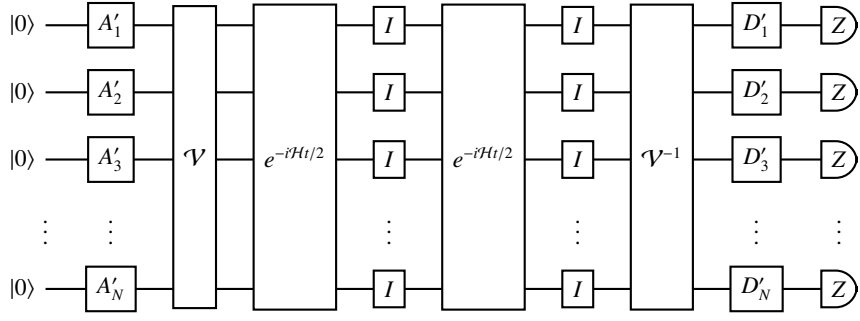

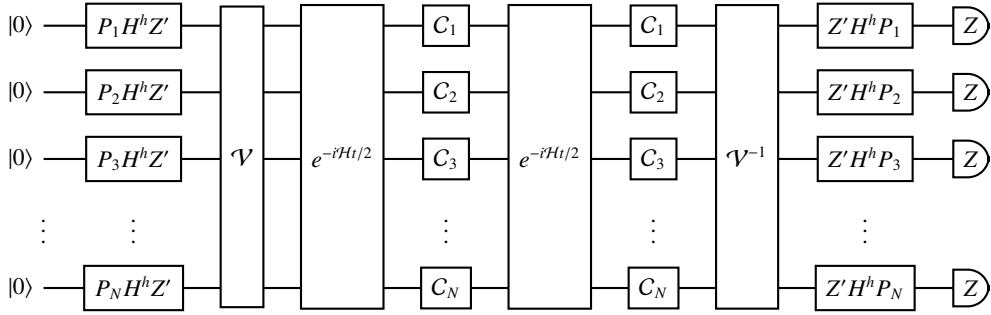
\begin{figure*}[t]
    \centering
 \begin{quantikz}[row sep=0.3cm]
% First qubit
\lstick{$\ket{0}$} & \gate{P_1 H^h Z^{\prime}} & \gate[wires=5, nwires=4]{\mathcal{V}} & \gate[wires=5, nwires=4]{e^{-i \mathcal{H} t / 2}} & \gate{\mathcal{C}_1} & \gate[wires=5, nwires=4]{e^{-i \mathcal{H} t / 2}} & \gate{\mathcal{C}_1} & \gate[wires=5, nwires=4]{\mathcal{V}^{-1}} &  \gate{Z^{\prime} H^h P_1} &\meterD{Z}\\
% Second qubit
\lstick{$\ket{0}$} & \gate{P_2 H^h Z^{\prime}} & & & \gate{\mathcal{C}_2} & &  \gate{ \mathcal{C}_2} & & \gate{Z^{\prime} H^h P_2} & \meterD{Z}\\
% Third qubit
\lstick{$\ket{0}$} & \gate{P_3 H^h Z^{\prime}} & & &\gate{\mathcal{C}_3} &  &  \gate{ \mathcal{C}_3} & &  \gate{Z^{\prime} H^h P_3} & \meterD{Z}\\
\vdots & \vdots &  & & \vdots& & \vdots & &\vdots \\
% Fourth qubit
\lstick{$\ket{0}$} & \gate{ P_N H^h Z^{\prime}} & & &\gate{\mathcal{C}_N} &  &  \gate{ \mathcal{C}_N} & & \gate{Z^{\prime} H^h P_N} & \meterD{Z}
\end{quantikz}
    \caption{Trap simulation with encoding of initial state, $\mathcal{V}$, and decoding after the time evolution, $\mathcal{V}^{-1}$. 
    $H$ denotes the Hadamard gate, $h \in \{0,1\}$ is a random bit, $Z^{\prime}$ denotes applying a Pauli Z gate with probability $0.5$ (each instance of random operator $Z^{\prime}$ inside the circuit is independent)  and
  $\otimes_{j = 1}^n \mathcal{C}_j = \mathcal{C}$ is the time inversion circuit,
  $\mathcal{H}$ is an arbitrary accreditable Hamiltonian, $P_j$ is a single-qubit Pauli gate, chosen uniformly at random independently for each $j$, and $t$ is the duration of the simulation. Note that $P_jH^h\bar{Z}^{\prime}$'s (where $Z^{\prime}$ is first and $H^h$ is second in time order, but in the opposite order when written as operators) are in the same box because they can be compiled as a single gate.
  } \label{fig:PauliTwirledTrapCircuit(withHadamardsandConglom)WITHV}
\end{figure*}

\begin{lemma}
\label{ProbCompare}
If E3 holds, the variational distance between the probability distribution over measurement outcomes of an error-free implementation and that of the erroneous implementation is greater in a trap simulation, for any value of the random parameters, than in the target simulation.
 \end{lemma}
 \begin{proof}
    The trap simulations are constructed so that traps can be obtained from the target by adding single-qubit gates to the target simulation in place of identity gates.
    E3 then implies that the variational distance
between the probability distribution over measurement outcomes of an error-free implementation and that of the erroneous implementation is greater in a trap simulation than in the target simulation.
 \end{proof}

\begin{lemma}[Detection of errors]
    \label{stochDetectLemma}
   Any error, or combination of errors, obeying E1-E3 and occurring within a trap simulation are detected with a probability of at least $1/2$, unless the errors cancel with each other.
\end{lemma}

The proof is provided in Appendix~\ref{sec:lemma3proof}. 

\begin{theorem}
\label{finalTheorem}
    Protocol \ref{fullProtocolSketch} performs accredited analogue simulation as per Definition \ref{AAQSDef}
     subject to the the error model, i.e.,  E1-E3, with $N_{\textrm{tr}}$ trap simulations, where
    \be
        N_{\mathrm{tr}} = \bigg \lceil \dfrac{2}{\theta^2} \ln{\bigg( \dfrac{2}{1 - \alpha} \bigg)} \bigg \rceil + 1
    \ee
\end{theorem} 

This is our central result. Crucially, the additional resources required for our quantum accreditation protocol are independent of the size ($N$) as well as the specifics
(inputs in Definition~\ref{AQSDef}) of the analogue quantum simulation. The proof is provided in Appendix~\ref{sec:proofs}.

\subsection{Design of Trap Simulation}
\label{sec:trapdesign}

This section provides fuller details of the trap simulation used in Sec. \ref{protocolSubsection}. 

\subsubsection{Time Inversion of Accreditable Hamiltonian}
Our trap circuits use the notion of time-inversion circuits that effectively inverts the time evolution of the Hamiltonian. 
We show that such a circuit exists for accreditable Hamiltonians on a large class of lattices, of which the square in an instance.

\begin{definition}
\label{timeReverseCircDef}
For a specific accreditable Hamiltonian, $\mathcal{H}$, a time-inversion circuit, $\mathcal{C}$, is an operator such that,
\begin{align}
    \mathcal{C} \mathcal{H} \mathcal{C}^{\dag} = - \mathcal{H}.
\end{align}
We refer to this as inverting the Hamiltonian.
\end{definition}

A circuit, $\mathcal{C}$, conforming to Definition \ref{timeReverseCircDef} suffices to reverse the time evolution of an accreditable Hamiltonian due to Lemmas \ref{unitExp} and \ref{InvertTime}.
\begin{lemma}\( \\ \)
\label{unitExp}
    For any unitary, $\mathcal{U}$, Hamiltonian, $\mathcal{H}$, $\kappa \in\mathbb{C}$ and $t \in \mathbb{R}$:
    \begin{align}
        \mathcal{U} e^{-\kappa \mathcal{H} t} \mathcal{U}^{\dag}
        &=
        e^{- \kappa \mathcal{U} \mathcal{H}\mathcal{U}^{\dag} t}
    \end{align}
\end{lemma}
\noindent Lemma \ref{unitExp} is proven in Appendix~\ref{InvertXYSubsection}.

\begin{lemma}[Inverting the Hamiltonian] Given $\mathcal{C},$ $\mathcal{H}$ as in Definition~\ref{timeReverseCircDef},
    \label{InvertTime}
    \begin{align}
    \mathcal{C} e^{-i \mathcal{H} t} \mathcal{C}^{\dag}
    &=
    e^{i  \mathcal{H} t}
\end{align} 
\end{lemma}

\begin{proof}
    Via Lemma \ref{unitExp},
    \begin{align}
        \mathcal{C} e^{-i \mathcal{H} t} \mathcal{C}^{\dag}
        &=
        e^{-i \mathcal{C} \mathcal{H}\mathcal{C}^{\dag} t}
    \end{align}  
    Definition \ref{timeReverseCircDef} then implies the lemma.
\end{proof}

The existence of a time inversion circuit, meeting the requirements of Definition \ref{timeReverseCircDef} for an accreditable Hamiltonian, $\mathcal{H}$
is established by Theorem \ref{fullInversionXYTheorem}.
\begin{theorem}\label{fullInversionXYTheorem}
    For any set of XY-interactions, where the interactions and qubits form a two-colourable graph with the qubits as vertices and interactions as edges, the corresponding accreditable Hamiltonian can be inverted by applying a time inversion circuit consisting of Pauli $Z$ gates on a chromatic subset (as defined in Definition~\ref{cromatSetDef}) of the qubits. 
\end{theorem}
The proof is provided in Appendix \ref{proofOfTimeReversal}.
Some examples of the use of time-inversion circuits to invert time evolutions are given in Appendix \ref{ExamplesAppendix}.

\subsubsection{Traps in the error-free case}
\label{error-freeTrapReduceSubsection}

We now present the measurement statistics of the traps in the error-free case. The erroneous case is more involved and hence is presented in Appendix~\ref{sec:lemma3proof}. The measurement statistics are in fact quite simple: the traps return a known `correct' result when implemented without any error occurring. This makes use of the time inversion circuits displayed above and is demonstrated in Lemma \ref{error-freeTrapReduce}. 
\begin{lemma}
    \label{error-freeTrapReduce}
    
    The error-free implementation of a trap simulation (on $N$ qubits) always gives the all-zero output with certainty.
\end{lemma}
\begin{proof}
    If no errors occur, the trap simulation, as in Fig. \ref{fig:PauliTwirledTrapCircuit(withHadamardsandConglom)WITHV} gives the all-zero output with probability
    \begin{align}
        \label{firstMathEquivTrap}
        \vert \langle 0 \vert^{\otimes N} \bar{Z}^{\prime} H^h \mathcal{P} \mathcal{V}^{-1} \mathcal{C} e^{-i \mathcal{H} t /2}  \mathcal{C} e^{-i \mathcal{H} t /2} \mathcal{V} \mathcal{P} H^h \bar{Z}^{\prime} \vert 0 \rangle^{\otimes N} \vert^2,
    \end{align}
    where $\bar{Z}^{\prime}$ denotes a Pauli $Z$ gate on each qubit, with probability $1/2$ (each qubit choice is independent so all possible combinations of Pauli $Z$ gates or the identity, e.g. $IZI \cdots IZ$, on $N$ qubits occurs with equal probability when a $\bar{Z}^{\prime}$ gate is implemented), and $\mathcal{P}$ is a independent uniformly random string of Pauli (or identity) gates ( one on each qubit).\\
Then, Lemma \ref{InvertTime} allows us to re-write the quantity in Eqn.~\ref{firstMathEquivTrap} as
\begin{align}
    & \vert \langle 0 \vert^{\otimes N} \bar{Z}^{\prime} H^h \mathcal{P}\mathcal{V}^{-1} e^{i \mathcal{H} t /2}  e^{-i \mathcal{H} t /2} \mathcal{V} \mathcal{P} H^h \bar{Z}^{\prime}\vert 0 \rangle^{\otimes N} \vert^2\\
    &=
    \label{timeEvolsCancel}
    \vert \langle 0 \vert^{\otimes N} \bar{Z}^{\prime} H^h \mathcal{P} \mathcal{V}^{-1} \mathcal{V} \mathcal{P} H^h \bar{Z}^{\prime} \vert 0 \rangle^{\otimes N} \vert^2
    =
    1.
\end{align}
\end{proof}

% Full Algorithm Sketch
\begin{figure*}
    \centering
\begin{algorithm}[H]
%\SetAlgoLined

% Changes the the 'Algorithm' to a 'Protocol'
\SetAlgorithmName{Protocol}{protocol}{List of Protocols}

$\mathbf{Input:}$ \\
$\bullet$ A Hamiltonian, $\mathcal{H}_0$\\
$\bullet$ A real evolution time, $t \in \mathbb{R}$\\
$\bullet$ An initial state, $\vert \psi_0 \rangle$\\ 
$\bullet$ A vector of measurements for after the time evolution, $\mathcal{M}$\\
$\bullet$ A desired confidence, $\alpha$.\\
$\bullet$ A desired accuracy of the output bound, $\theta$.\\
\vspace{0.1cm}
 \begin{enumerate}

    \item Reduce $\mathcal{H}_0$ to an accreditable Hamiltonian, $\mathcal{H}$
    \item Calculate the circuit, $\mathcal{V}$, that transforms $\vert 0 \rangle^{\otimes N}$ to the initial state corresponding to $\vert \psi_0 \rangle$ after the reduction to an accreditable Hamiltonian. 
     \item Calculate the required number of traps $N_{\textrm{tr}} = \bigg \lceil \dfrac{2}{\theta^2} \ln{\bigg( \dfrac{2}{1 - \alpha} \bigg)} \bigg \rceil + 1$.
     \item Pick uniformly at random a number between $1$ and $N+1$ to be the index of the target simulation 
     \item For $i=1$ to $N_{\textrm{tr}}+1$
     \begin{enumerate}
         \item If simulation $i$ is the target simulation:
         
    \hspace{1.5em}Prepare the state $\vert 0 \rangle^{\otimes N}$.\\
    \hspace{1.5em} For qubit $j$:\\
    \hspace{3.0em}Apply $A^{\prime}_j$\\
    \hspace{1.5em}Apply $\mathcal{V}$\\
    \hspace{1.5em}Apply $e^{-i \mathcal{H} t / 2}$.\\
    \hspace{1.5em}Apply Identity.\\
    \hspace{1.5em}Apply $e^{-i \mathcal{H} t / 2}$.\\
    \hspace{1.5em}Apply Identity.\\
    \hspace{1.5em}Apply $\mathcal{V}^{-1}$\\
    \hspace{1.5em} For qubit $j$:\\
    \hspace{3.0em}Apply $D^{\prime}_j$\\
    \hspace{1.5em}Perform measurements \& record outcomes.
    
        \item If simulation $i$ is a trap simulation:

    \hspace{1.5em} For qubit $j$:\\
    \hspace{3.0em} Choose Pauli gate $P_j$ uniformly at random\\
    \hspace{1.5em} Choose $h \in \{ 0,1 \}$ uniformly at random\\
    \hspace{1.5em} Prepare the state $\vert 0 \rangle^{\otimes N}$.\\ 
    \hspace{1.5em} For qubit $j$:\\
    \hspace{3.0em} With probability $0.5$ apply a Pauli $Z$ gate to qubit $j$\\
    \hspace{3.0em} Apply $P_jH^h$ to qubit $j$\\
    \hspace{1.5em} Apply $\mathcal{V}$\\
    \hspace{1.5em} Apply $e^{-i \mathcal{H} t / 2}$\\
    \hspace{1.5em} Apply time inversion circuit\\
    \hspace{1.5em} Apply $e^{-i \mathcal{H} t / 2}$\\
    \hspace{1.5em} Apply time inversion circuit\\
    \hspace{1.5em} Apply $\mathcal{V}^{-1}$\\
    \hspace{1.5em} For qubit $j$:\\
    \hspace{3.0em} Apply $H^hP_j$ to qubit $j$\\
    \hspace{3.0em} With probability $0.5$ apply a Pauli $Z$ gate to qubit $j$\\
    \hspace{1.5em} Perform measurements \& record if correct
     \end{enumerate}

\item Calculate   $\epsilon
    = 
    2\dfrac{\text{Number of correct traps}}{\text{Total number of traps}}$
\end{enumerate}

\vspace{0.1cm}
$\mathbf{Return}:$ Target simulation measurement outcome sample and $\epsilon$.
\caption{Analogue Accreditation Protocol
 \label{fullProtocolSketch}}
\end{algorithm}
\end{figure*}

\section{Discussion}
\label{discuss}
   We have presented a quantum accreditation protocol for quantum analogue simulations that can be applied to extant experiments and devices. 	
   	It builds on the theoretical advances of strongly universal Hamiltonians and quantum accreditation as well as 
   experimental progress towards the realisation of programmable hybrid analogue–digital quantum evolution.
   
      {Our protocol} completely eliminates the need for classical simulations, freeing us to accredit simulations of arbitrarily large systems where quantum simulators offer the most.
   Our error model captures large classes of errors hybrid quantum simulators experience.
   Additionally, the resource requirements of our protocol are reasonable: 
   the depth and time overheads are independent of the size of the system being simulated ($N$); 
   the number of extra single-qubit gates required is at most linear in the system size; 
   the total duration the time evolution is applied for remains unchanged from an un-accredited simulation; 
   and the number of trap simulations required is quadratic in the reciprocal of the required accuracy of the bound on the variational distance the protocol outputs.
  
     Consequently, our protocol can be implemented on extant programmable hybrid analogue-digital quantum simulators~\cite{Browaeys2020,Bluvstein2022}.
      It is particularly amenable if the HQS implements the XY-interaction on a 2-colourable graph as it eliminates the $\mathcal{V}$ and $\mathcal{V}^{-1}$ operations in Figs.~\ref{fig:targInSplitUpWITHV} and~\ref{fig:PauliTwirledTrapCircuit(withHadamardsandConglom)WITHV}.

   Our work leaves several potential avenues for improvement, centred particularly around relaxing E2 and E3 of our error model.
   The independence assumed in E2 contributes, via Lemma~\ref{simplerErrorLem}, to the error in HQS being independent of the single-qubit gates $A_j, B_j, D_j.$ 
   E2 may be relaxed to allow for error that depends weakly on single-qubit-gates~\cite[Appendix 2]{2021Sams}.
   This may be combined with a relaxed identicality assumption as well by using probability concentration inequalities more permissive than Hoeffding's inequality. 
   
   Relaxing E3 would require understanding it better in terms of its physical implications, as it is the most novel and least explored of our assumptions. 
   Its relaxation may benefit from inverting the Hamiltonian more than once.
   This would however increase the overheads of quantum accreditation in terms of time and single-qubit gates.
   
   Finally, given the trap-based nature of our accreditation protocol, it may be tempting to suggest an `error' wherein the accreditable Hamiltonian 
   $\mathcal{H}$ in Figs. ~\ref{fig:targInSplitUpWITHV} and~\ref{fig:PauliTwirledTrapCircuit(withHadamardsandConglom)WITHV} is replaced by another $\mathcal{H'}.$
   This error will cancel in the trap simulation, thus effecting an error that our trap simulation seemingly fails to detect. 
   As its effect will not cancel in the target simulation, this `error' violates E3 and is mathematically disallowed by our error model.
   Physically, the replacement of $\mathcal{H}$ by $\mathcal{H'}$ \emph{ceteris paribus} is unlikely to be due to noise or stochastic miscalibrations.

\section{Acknowledgements}

We thank Ross Grassie, Sean Thrasher, James Mills, and Raul Garcia-Patron for useful conversations.
This work was supported, in part, 
by the UK Networked Quantum Information Technologies (NQIT) Hub (EP/M013243/1), 
the UKRI ExCALIBUR project QEVEC (EP/W00772X/2), 
and a Leverhulme Trust Early Career Fellowship.\\

\bibliography{References}

%%%%%%%%%%%%%%%%%%%%%%%%%%%% APPENDIX %%%%%%%%%%%%%%%%%%%%%

\onecolumngrid
\appendix

\section{Proof of Lemma~\ref{simplerErrorLem}}
\label{SimplifyingErrorApp}
\begin{proof}[Proof of Lemma~\ref{simplerErrorLem}]
Consider the operations in the circuit of Fig.~\ref{fig:generalFormForErrorfulSinglesRepeat1}. Each operation is affected by error that obeys E1-E3 (including preparation and measurement errors). Each error is a CPTP map acting before or after the operation. Let $\mathcal{E}_O$ denote the CPTP map representing the error in operation $O$, so that the action of the erroneous device can be written as $\mathcal{E}_O \left( (O \otimes I) \sigma  (O^{\dagger} \otimes I)  \right)$, where $\sigma$ is a density matrix representing the state of the HQS+environment just before the application of $O$, and $I$ an identity operator applying on the environment.

Denote $\bar{A} \equiv A_1 \otimes A_2 \otimes \ldots A_N$, and $\bar{B}$, $\bar{C}$, and $\bar{D}$ defined similarly. Also let $\mathcal{E}_{\mathcal{H}, t}^{\prime}$, $\mathcal{E}_{\mathcal{H}, t}$ represent the CPTP error in the first and second time evolutions, respectively.

For the rest of the proof, by abuse of notation, we will use the same symbol for a CPTP map that has a single unitary Kraus operator and the unitary operator itself. For example $\bar{A}$ will be used to denote CPTP map with unitary Kraus operator  $\bar{A}$. Similarly $e^{-i \mathcal{H} t/2}$ will be used to denote CPTP map with unitary Kraus operator $e^{-i \mathcal{H} t/2}$. We will then use composition of CPTP maps to proceed with our proof.

Then, the erroneous implementation of the operations in the circuit of Fig.~\ref{fig:generalFormForErrorfulSinglesRepeat1}, not including the preparation and measurement operations, on system+environment can be expressed mathematically (operation is map composition):

\ben
    && \bar{D} \mathcal{E}_{\bar{D}} \mathcal{U}_2 \mathcal{E}_{\mathcal{U}_2} \bar{C}  \mathcal{E}_{\bar{C}} \mathcal{E}_{\mathcal{H}, t} 
    e^{-i \mathcal{H} t/2} \bar{B} \mathcal{E}_{\bar{B}} \mathcal{E}^{\prime}_{\mathcal{H}, t} e^{-i \mathcal{H} t/2} \mathcal{E}_{\mathcal{U}_1} \mathcal{U}_1 \mathcal{E}_{\bar{A}} \bar{A}\\
    &=&\bar{D} \big( \mathcal{U}_2 \bar{C} \bar{C}^{\dagger} \mathcal{U}_2^{\dagger} \big) \mathcal{E}_{\bar{D}} \mathcal{U}_2 \mathcal{E}_{\mathcal{U}_2} \bar{C} \mathcal{E}_{\bar{C}} \mathcal{E}_{\mathcal{H}, t} e^{-i \mathcal{H} t/2} \bar{B} \mathcal{E}_{\bar{B}} \mathcal{E}^{\prime}_{\mathcal{H}, t} e^{-i \mathcal{H} t /2} \mathcal{E}_{\mathcal{U}_1} \mathcal{U}_1 \mathcal{E}_{\bar{A}} \big( \mathcal{U}_1^{\dagger} e^{i \mathcal{H} t/2} e^{-i \mathcal{H} t/2} \mathcal{U}_1 \big) \bar{A} \\
     &=&
    \bar{D} \mathcal{U}_2 \bar{C} \big( \bar{C}^{\dagger} \mathcal{U}_2^{\dagger} \mathcal{E}_{\bar{D}} \mathcal{U}_2 \mathcal{E}_{\mathcal{U}_2} \bar{C} \mathcal{E}_{\bar{C}} \mathcal{E}_{\mathcal{H}, t} \big) e^{-i \mathcal{H} t/2} \bar{B} \big( \mathcal{E}_{\bar{B}} \mathcal{E}^{\prime}_{\mathcal{H}, t} e^{-i \mathcal{H} t /2} \mathcal{E}_{\mathcal{U}_1} \mathcal{U}_1 \mathcal{E}_{\bar{A}} \mathcal{U}_1^{\dagger} e^{i \mathcal{H} t/2} \big) e^{-i \mathcal{H} t/2} \mathcal{U}_1 \bar{A} \\
    &=&
       \bar{D} \mathcal{U}_2 \bar{C} \left( \Tilde{\mathcal{E}}_{\mathcal{H}, t} e^{-i \mathcal{H} t/2} \right) 
       \bar{B}  \left(\Tilde{\mathcal{E}}^{\prime}_{\mathcal{H}, t}e^{-i \mathcal{H} t/2} \right)
        \mathcal{U}_1 \bar{A},
        \label{eq:lemma1}
\een
 where we have defined
 $\Tilde{\mathcal{E}}_{\mathcal{H}, t} = \bar{C}^{\dagger} \mathcal{U}_2^{\dagger} \mathcal{E}_{\bar{D}} \mathcal{U}_2 \mathcal{E}_{\mathcal{U}_2} \bar{C} \mathcal{E}_{\bar{C}} \mathcal{E}_{\mathcal{H}, t}$ and $\Tilde{\mathcal{E}}^{\prime}_{\mathcal{H}, t} = \mathcal{E}_{\bar{B}} \mathcal{E}^{\prime}_{\mathcal{H}, t} e^{-i \mathcal{H} t /2} \mathcal{E}_{\mathcal{U}_1} \mathcal{U}_1 \mathcal{E}_{\bar{A}} \mathcal{U}_1^{\dagger} e^{i \mathcal{H} t/2}.$

 Due to the assumption of independence in E2, the error in $\bar{A}$, $\bar{B}$ and $\bar{D}$ is gate-independent, 
 and $\Tilde{\mathcal{E}}_{\mathcal{H}, t}$ and $\Tilde{\mathcal{E}}^{\prime}_{\mathcal{H}, t}$ are independent of $\bar{A}$, $\bar{B}$ and $\bar{D}$.
Hence the erroneous implementation of any element of Fig.~\ref{fig:generalFormForErrorfulSinglesRepeat1}, 
barring the preparation and measurement,
can be considered as if all operations except the time evolutions are error-free 
and the error in the time evolutions is represented by an operator independent of $\bar{A}$, $\bar{B}$ and $\bar{D}$, as noted within the parentheses in Eqn.~\ref{eq:lemma1}.
\end{proof}

\section{Proof of Lemma \ref{stochDetectLemma}}
\label{sec:lemma3proof}

As in Appendix~\ref{SimplifyingErrorApp}, by abuse of notation, we will use the same symbol for a CPTP map that has a single unitary Kraus operator and the unitary operator itself. All operations that follow are compositions of CPTP maps. We first prove an auxiliary lemma.

\begin{lemma}
\label{errorConciseFormatTrap}
    Presupposing E1, 
    the probability of measuring a state  $\vert \phi \rangle \in \big \{ \vert 0 \rangle, \vert 1 \rangle \big\}^N$ in an erroneous implementation of a trap simulation on $N$ qubits is
    \begin{align}
    \label{ErrorTrapAnalysisLemmaStatement}
    \tr\left(\ket{\phi}\bra{\phi}  \mathcal{E}_{\mathrm{meas}} \bar{Z}^{\prime} \bar{H}^h \mathcal{P} \mathcal{E}_{\mathrm{com}} \mathcal{P} \bar{H}^h \bar{Z}^{\prime} \mathcal{E}_{\mathrm{prep}} \left( \ket{0}\bra{0}^{\otimes N} \otimes \ket{env}\bra{env} \right)  \right),
\end{align}
where $\mathcal{E}_{\mathrm{prep}}, \mathcal{E}_{\mathrm{meas}}, \mathcal{E}_{\mathrm{com}}$ are CPTP maps independent of $h$ and $\mathcal{P}$, $\bar{H}=H^{\otimes N}$, $\mathcal{P}=P_1 \otimes P_2 \otimes \ldots \otimes P_N$,
and $\bar{Z}^{\prime}$ is the application with probability $1/2$ of a Pauli $Z$ gate independently on each qubit, $\ket{env}$ is the initial state of the environment and the rest of notation is the same as in Appendix~\ref{SimplifyingErrorApp}.
\end{lemma}
\begin{proof}

Via Lemma \ref{simplerErrorLem}, we can consider a trap simulation with error only in the time evolutions, state preparation, and measurement when considering trap simulations with error. The erroneous implementation gives a probability of measuring $\vert \phi \rangle$ of
\begin{align}
    \label{ErrorTrapAnalysisFirst}
    \tr \left( \ket{\phi}\bra{\phi} \mathcal{E}_{\mathrm{meas}} \bar{Z}^{\prime} \bar{H}^h \mathcal{P} \mathcal{V}^{-1} \mathcal{C} \Tilde{\mathcal{E}}_{\mathcal{H}, t} e^{-i \mathcal{H} t /2}  \mathcal{C} \Tilde{\mathcal{E}}^{\prime}_{\mathcal{H}, t} e^{-i \mathcal{H} t /2} \mathcal{V} \mathcal{P} \bar{H}^h \bar{Z}^{\prime} \mathcal{E}_{\mathrm{prep}} \left( \ket{0}\bra{0}^{\otimes N} \otimes \ket{env}\bra{env} \right) \right),
\end{align}
where the $\Tilde{\mathcal{E}}_{\mathcal{H}, t}, \Tilde{\mathcal{E}}^{\prime}_{\mathcal{H}, t}$, $\mathcal{E}_{\mathrm{prep}}$, $\mathcal{E}_{\mathrm{meas}}$ are CPTP maps (representing error) that are independent of the value of $h$ and the choice of $\mathcal{P}$ because of Lemma~\ref{simplerErrorLem}.
$\mathcal{C}$ denotes the time inversion circuits, as in the definition of the trap.\\

Using  $\mathcal{U} \mathcal{E} = \mathcal{E'} \mathcal{U}$, where $\mathcal{E'}=\mathcal{U} \mathcal{E} \mathcal{U}^{-1}$ and $\mathcal{U}$ an invertible CPTP map, we can commute $\Tilde{\mathcal{E}}_{\mathcal{H}, t}$ to the left, in Eqn.~\ref{ErrorTrapAnalysisFirst}, until it meets the leftmost $\mathcal{P}$; similarly, commuting $\Tilde{\mathcal{E}}^{\prime}_{\mathcal{H}, t}$ to the right, in Eqn.~\ref{ErrorTrapAnalysisFirst}, until it meets the rightmost $\mathcal{P}$, gives,
\begin{align}
    \label{ErrorTrapAnalysisSecond}
    \tr \left( \ket{\phi} \bra{\phi} \mathcal{E}_{\mathrm{meas}} \bar{Z}^{\prime} \bar{H}^h \mathcal{P} \hat{\mathcal{E}}_{\mathcal{H}, t} \mathcal{V}^{-1} \mathcal{C} e^{-i \mathcal{H} t /2}  \mathcal{C} e^{-i \mathcal{H} t /2} \mathcal{V} \hat{\mathcal{E}}^{\prime}_{\mathcal{H}, t} \mathcal{P} \bar{H}^h \bar{Z}^{\prime} \mathcal{E}_{\mathrm{prep}} \left( \ket{0}\bra{0}^{\otimes N}  \otimes \ket{env}\bra{env} \right) \right),
\end{align}
for some other CPTP maps, $\hat{\mathcal{E}}_{\mathcal{H}, t}$ and $\hat{\mathcal{E}}^{\prime}_{\mathcal{H}, t}$, that are also independent of $h$ and $\mathcal{P}$. 
As in the error-free case in Sec. \ref{error-freeTrapReduceSubsection}, there are many cancellations. So Eqn.~\ref{ErrorTrapAnalysisSecond} is mathematically equivalent to,
\begin{align}
    \label{ErrorTrapAnalysisThird}
     \tr \left( \ket{\phi} \bra{\phi} \mathcal{E}_{\mathrm{meas}} \bar{Z}^{\prime} \bar{H}^h \mathcal{P} \mathcal{E}_{\mathrm{com}} \mathcal{P} \bar{H}^h \bar{Z}^{\prime} \mathcal{E}_{\mathrm{prep}} \left( \ket{0}\bra{0}^{\otimes N} \otimes \ket{env}\bra{env} \right) \right),
\end{align}
where $\mathcal{E}_{\mathrm{com}} = \hat{\mathcal{E}}_{\mathcal{H}, t} \hat{\mathcal{E}}^{\prime}_{\mathcal{H}, t}$ is a CPTP map independent of $h$ and $\mathcal{P}$.
\end{proof}

\begin{proof}[Proof of Lemma ~\ref{stochDetectLemma}]
 Lemma~\ref{errorConciseFormatTrap} allows us to express the probability of measuring a state, $\vert \phi \rangle \in \big \{ \vert 0 \rangle, \vert 1 \rangle \big\}^{\otimes N}$ from in a trap as,

\begin{align}
    \label{ErrorTrapAnalysisRestateInLemma3}
     \tr \left( \ket{\phi} \bra{\phi} \mathcal{E}_{\mathrm{meas}} \bar{Z}^{\prime} \bar{H}^h \mathcal{P} \mathcal{E}_{\mathrm{com}} \mathcal{P} \bar{H}^h \bar{Z}^{\prime} \mathcal{E}_{\mathrm{prep}} \left( \ket{0}\bra{0}^{\otimes N} \otimes \ket{env}\bra{env} \right) \right),
\end{align}
where $\mathcal{E}_{\mathrm{prep}}, \mathcal{E}_{\mathrm{meas}}, \mathcal{E}_{\mathrm{com}}$ are CPTP maps independent of the single-qubit gates stated in Lemma~\ref{errorConciseFormatTrap}.
The uniformly random Pauli operators now surrounding the CPTP maps in Eqn.~\ref{ErrorTrapAnalysisRestateInLemma3} implement a Pauli twirl on their respective CPTP map, effectively reducing each error to stochastic Pauli error (the random Pauli operators are also effectively removed by this twirl).
 
More specifically, that the effect of the $Z^{\prime}$ gates in Eqn.~\ref{ErrorTrapAnalysisRestateInLemma3} are to effectively reduce the error represented by $\mathcal{E}_{\mathrm{prep}}$ and $ \mathcal{E}_{\mathrm{meas}}$ to stochastic Pauli $X$ gate error, and the effect of the $\mathcal{P}$ (which are random Pauli gates on each qubit in the simulation) is to twirl the error represented by $\mathcal{E}_{\mathrm{com}}$ to stochastic Pauli error~\cite[Appendix A]{2019Sam}. We then update Eqn.~\ref{ErrorTrapAnalysisRestateInLemma3}, setting $\ket{\phi}$ to $\ket{0}^{\otimes N}$ (as this is the probability we are interested in):
\begin{align}
    \label{ErrorTrapAnalysisRestateInLemma3Second}
    \tr \left( \ket{0}\bra{0}^{\otimes N} \mathcal{E}_{\mathrm{meas}}' \bar{H}^h \mathcal{E}_{\mathrm{com}}' \bar{H}^h \mathcal{E}_{\mathrm{prep}}' \left( \ket{0}\bra{0}^{\otimes N} \otimes \ket{env}\bra{env} \right) \right)
\end{align}
Where $\mathcal{E}_{\mathrm{prep}}', \mathcal{E}_{\mathrm{meas}}', \mathcal{E}_{\mathrm{com}}'$ represent the post-Pauli twirl equivalents of $\mathcal{E}_{\mathrm{prep}}, \mathcal{E}_{\mathrm{meas}}, \mathcal{E}_{\mathrm{com}}$, respectively (i.e. as stochastic Pauli error).

We then consider the detection of each kind of Pauli error depending on whether the Hadamard gates (the $\bar{H}$ in Eqn.~\ref{ErrorTrapAnalysisRestateInLemma3Second}) are present or not\footnote{As $Y=iXZ$ stochastic Pauli error is treated as stochastic Pauli $X$ or $Z$ error.}.
Note that the Pauli error constituting $\mathcal{E}_{\mathrm{prep}}', \mathcal{E}_{\mathrm{meas}}'$ are always Pauli $X$ and so are always detected, if they occur; the presence of the Hadamard gates only affects whether the error in $\mathcal{E}_{\mathrm{com}}$ is detected or not.\

For the case of $h=0$, which happens with probability $1/2$, and results in the Hadamard gates not being present, Pauli $X$ errors flip the outcome of the Pauli $Z$ measurements and so are detected.\\

For the case of $h=1$, which also happens with probability $1/2$, and results in the Hadamard gates being present, the Pauli $Z$ error is mapped to Pauli $X$ error (and \textit{vice versa}) by being surrounded by the Hadamard gates. Therefore the Pauli $Z$ measurements detect what was previously (before the Hadamard gates switched the Pauli $X$ and $Z$ error) the Pauli $Z$ error but not what was Pauli $X$ i.e., any error that was not detected in the $h=0$ case.

Therefore, any stochastic Pauli error, as the Pauli twirls have effectively reduced all error to, is detected in either the $h=0$ case or the $h=1$ case (or both, if it is Pauli $Y$ error). Each case occurs with probability $1/2$, hence any error occurring is detected with probability at least $1/2$.

\end{proof}

\section{Proof of Theorem \ref{finalTheorem}}

\label{sec:proofs}.

\noindent We begin with a definition.

\begin{definition}
\label{BoundedFormDef}
An error occurring in a HQS is said to be in stochastic-error-form if there exists an $r \in [0,1]$ such that,
\be
    \label{boundedDefEquation}
    \Tilde{\rho}	=    \big( 1 - r\big) \rho +  r \sigma,
\ee
where
 $\rho$ is the state of the HQS if the error had never occurred,
 $\Tilde{\rho}$ is the state of the HQS right after the error occurs,
 $\sigma$ is an arbitrary state of the HQS that encompasses all errors.
\end{definition}

Given Lemma~\ref{ProbCompare} it suffices to acquire (experimentally) an estimate of a bound on the error of the traps in order to get a bound on the $\mathrm{VD}$ of the target simulation output from its error-free distribution. The bound on the error of the traps is acquired indirectly by estimating the probability trap simulations return an incorrect measurement outcome. Subsequently we get a bound on the $\mathrm{VD}$ of the target using the following Lemma.

\begin{lemma}
    \label{finalBound}
    For a target simulation of any accreditable Hamiltonian, if $p_{\mathrm{inco}}$ is the probability the corresponding trap simulations give an incorrect output, 
    and $\mathrm{VD}(P, \Tilde{P})$ is the variational distance between:\\
    $\bullet$ the probability distribution, $P$, over the measurement outcomes in the target circuit, assuming no error occurs\\ 
    $\bullet$ the experimentally-obtained probability distribution, $\Tilde{P}$, over the measurement outcomes in the target circuit when physically performed (where error may occur),\\
    then $\mathrm{VD}(P, \Tilde{P})$ is bounded as
    \begin{align}
        \mathrm{VD}(P, \Tilde{P}) \leq 2 p_{\mathrm{inco}}
    \end{align}
\end{lemma}
\begin{proof}
Lemma~\ref{ProbCompare} implies that the variational distance between the error-free and erroneous probability distributions is greater in the trap simulation than in the target simulation. Hence, the problem is reduced to showing that twice the probability a trap returns an incorrect output, $2 p_{\mathrm{inco}}$, is greater than the variational distance between the error-free and erroneous probability distributions in the trap simulations.\\
As all error in the trap simulations is Pauli twirled to stochastic Pauli error, it can be considered to be in stochastic-error-form (as defined in Def.~\ref{BoundedFormDef}).

If all the error is pushed to the very end of the HQS, then the state immediately before measurement (but after the CPTP map representing error in the measurement is applied) can be expressed as in Eqn.~\ref{boundedDefEquation} with $r \equiv p^{\text{Trap}}_{\mathrm{err}}.$

Following Ref.~\cite[Appendix Sec. 1]{2021Sams}:
\begin{align}
        \mathrm{VD}(P, \Tilde{P}) 
        \leq
        \dfrac{1}{2} \tr |\rho - \Tilde{\rho}|
        \leq
        \dfrac{p^{\mathrm{Trap}}_{\mathrm{err}}}{2} \tr | \rho - \sigma |
        \leq
        p^{\mathrm{Trap}}_{\mathrm{err}}.
    \end{align}
It thus suffices to prove that $p^{\mathrm{Trap}}_{\mathrm{err}} \leq 2 p_{\mathrm{inco}}$.
From Lemma~\ref{stochDetectLemma} any error occurring in a trap simulation is detected  with probability at least $1/2$, unless they cancel.\\
The cancellation of errors, in traps, is unaffected by the choice of random single-qubit gates (seen most simply using Lemma~\ref{simplerErrorLem} and noting that no random single qubit gates are between the two time evolutions), hence cancelling combinations of errors can be regarded as simply not occurring and their probability of occurring not contributing to $ p^{\text{Trap}}_{\mathrm{err}}$.\\  

Therefore, assuming error occurs, the probability it is detected, resulting in an incorrect output from a trap, is at least $1/2$. Hence,
    \begin{align}
    0.5 &\leq  \text{prob} \big( \text{trap output incorrect } \big \vert \text{ error occurs} \big) \\
       \Rightarrow  0.5 p^{\text{Trap}}_{\mathrm{err}}
        &\leq
        \text{prob} \big( \text{trap output incorrect } \big \vert \text{ error occurs} \big) p^{\mathrm{Trap}}_{\text{err}}\\
       \Rightarrow p^{\text{Trap}}_{\mathrm{err}}
        &\leq
        2p_{\mathrm{inco}}
    \end{align}
\end{proof}

Lemma \ref{finalBound} provides a bound on the variational distance between error-free and erroneous simulations, given an estimate of $p_{\mathrm{inco}}$.
This can only be done within a certain error, with preset confidence.

Due to E2, each run of a trap simulation is a Bernoulli trial with i.i.d. probability of detecting error: the trap either outputs the correct output, with probability (1 - $p_{\mathrm{inco}}$), or it does not with probability $p_{\mathrm{inco}}$.
The purpose of multiple runs of the trap is to estimate $p_{\mathrm{inco}}$ based on the statistics of the set of runs.
After $N_{\textrm{tr}}$ trap simulations, if $N_{\mathrm{inco}}$ is the number giving incorrect outputs, that is
 the number of traps where the measurements return the values known to be incorrect, 
then $p_{\mathrm{inco}}$ is approximately
\begin{align}
    p_{\mathrm{inco}} \approx \dfrac{N_{\mathrm{inco}}}{N_{\textrm{tr}}}.
\end{align}

\begin{proof}[Proof of Theorem~\ref{finalTheorem}]
    Protocol \ref{fullProtocolSketch} constructs trap and target simulations as defined in Sec. \ref{protocolSubsection}. E1-E3 enforce their correct functioning, as in Lemma \ref{stochDetectLemma}.
    E2 additionally implies that each run of a trap may be considered as an independent (of other runs) Bernoulli trial. 
    Protocol \ref{fullProtocolSketch} uses many traps and uses the number that give incorrect outputs to estimate the probability of error, $p_{\mathrm{inco}}$.\\
    Assuming there are $N_{\textrm{tr}}$ runs of traps and $N_{\mathrm{inco}}$ give incorrect outputs (i.e. detect error); define a random variable, $\mathbb{S} = \dfrac{N_{\mathrm{inco}}}{N_{\textrm{tr}}}$, which is both the value Protocol \ref{fullProtocolSketch} experimentally measures to estimate $p_{\mathrm{inco}}$ and also has the expected value, 
    \begin{align}
        \mathbb{E} \big[ \mathbb{S}\big] = p_{\mathrm{inco}}
    \end{align}
    For the specified allowed (with confidence $\alpha$) error, $\theta$, Hoeffding's inequality~\cite{doi:10.1080/01621459.1963.10500830} gives:
    \begin{align}
        \text{prob} \bigg( \bigg \vert 2p_{\mathrm{inco}} - 2 \mathbb{S} \bigg \vert \geq \theta \bigg) &\leq 2 e^{-N_{\textrm{tr}} \theta^2 / 2}
    \end{align}
    Using that $N_{\textrm{tr}} > \dfrac{2}{\theta^2} \ln{\bigg( \dfrac{2}{1 - \alpha} \bigg)}$, as assumed in the theorem statement,
    \begin{align}
        2 e^{-N_{\textrm{tr}} \theta^2 / 2} 
        &<
        1 - \alpha
    \end{align}
   Lemma~\ref{finalBound} then implies the result.
\end{proof}

\section{Proof of Theorem \ref{fullInversionXYTheorem}}
\label{proofOfTimeReversal}

\subsection{Preliminary Lemmas}

\begin{lemma}\( \\ \)
\label{CevitaLemma}
Let $\epsilon_{i, j, k}$ be the Levi-Cevita symbol.
If $p, j, k \in \{1, 2, 3\}$ and $p \not = j \not = k$, then
    $\epsilon_{k,j,p} \epsilon_{j, p, k} = 1$
\end{lemma}
\begin{proof} \( \\ \)
    As, by definition, $p \not = j, k$  and $j \not = k$,
    \begin{align}
        \epsilon_{k,j,p}, \epsilon_{j, p, k} \in \{ 1, -1\}
    \end{align}
    Then, by the rules of Levi-Cevita symbols:
    \begin{align}
        \epsilon_{k,j,p} 
        &=
        \epsilon_{j,p,k}
    \end{align}
    Therefore, as $\epsilon_{k,j,p} \in \{ 1, -1\}$,
    \begin{align}
        \epsilon_{k,j,p} \epsilon_{j, p, k} 
        &= 1
    \end{align}
\end{proof}
Lemma \ref{CevitaLemma} is then used in the proof of Lemma \ref{negationTheorem} below.
\begin{lemma}\( \\ \)
\label{negationTheorem}
    If any Pauli matrix is conjugated by any Pauli matrix, it gains a factor of $-1$ if the Pauli matrices are different, and is unaffected if they are the same.
\end{lemma}
\begin{proof}
    Let $\hat{\sigma}_j$, $\hat{\sigma}_k$ be Pauli matrices, then, by multiplying the Pauli matrices with the standard formula
\be
    \hat{\sigma}^{\dag}_j \hat{\sigma}_k \hat{\sigma}_j
   =
    \delta_{k,j} \hat{\sigma}_j + i \epsilon_{k,j,p} \bigg( \delta_{j,p} \mathcal{I} + i \epsilon_{j, p, q} \hat{\sigma}_q \bigg),
\ee
where $p \not = k,j$ and $q \not = j,p$.
Therefore, $q = k$.
Using these requirements, we further derive:
\begin{align}
\label{beforeLemma}
    \hat{\sigma}^{\dag}_j \hat{\sigma}_k \hat{\sigma}_j
    &=
    \delta_{k,j} \hat{\sigma}_j + i^2 \epsilon_{k,j,p} \epsilon_{j, p, k} \hat{\sigma}_k
\end{align}
Using Lemma \ref{CevitaLemma} in Eqn.~\ref{beforeLemma},
\begin{align}
    \hat{\sigma}^{\dag}_j \hat{\sigma}_k \hat{\sigma}_j
    &=
    \label{longFormPauliConjProduct}
    \delta_{k,j} \hat{\sigma}_j - \big( 1 - \delta_{k,j} \big) \hat{\sigma}_k
\end{align}
Therefore, by re-expressing Eqn.~\ref{longFormPauliConjProduct} more succinctly,
\begin{align}
    \hat{\sigma}^{\dag}_j \hat{\sigma}_k \hat{\sigma}_j &=
    \big( -1 \big)^{1 + \delta_{j,k}} \hat{\sigma}_k
\end{align}
Hence if conjugating Pauli matrix, $\hat{\sigma}$, is the different to the original Pauli matrix, $\hat{\sigma}_k$, then $\hat{\sigma}_k$ gains a factor of $-1$. Otherwise, nothing changes.
\end{proof}

\begin{definition}{$\sigma_k$-exclusion qubit}\( \\ \)
    A Hamiltonian, $\mathcal{H}$, expressed solely in terms of Pauli matrices has a $\sigma_k$-exclusion qubit if there is a qubit, q, such that every term in the Hamiltonian acts non-trivially on q but no term has the specific Pauli, $\sigma_k$, acting on q.
\end{definition}

\begin{lemma}\( \\ \)
    \label{HamiltonianPhase}
    If a Hamiltonian, $\mathcal{H}$, (built from Pauli operators) has a $\sigma_k$-exclusion qubit, assume this qubit has index $j$, for some Pauli matrix, $\sigma_k$, acting on qubit $j$ then,
    \begin{align}
        \hat{\sigma}_k^{\dag} \mathcal{H}\hat{\sigma}_k
        &=
        - \mathcal{H}
    \end{align}
    \underline{Example}\\
    The Hamiltonian $\mathcal{H} = X_1 X_2 + Y_1 Y_2$ is inverted as above when $\hat{\sigma}_k = Z_1$, or $\hat{\sigma}_k = Z_2$
\end{lemma}
\begin{proof}
$\mathcal{H}$ has a $\sigma_k$-exclusion qubit, by assumption, and assume this is the qubit with index $j$. Hence if $\mathcal{H}$ is conjugated by $\hat{\sigma}_k$  acting on qubit $j$, this can be rewritten as for each term in $\mathcal{H}$ (by linearity), $\hat{\sigma}_k$ and $\hat{\sigma}^{\dag}_k$ coming from each side of the term, commuting with every Pauli operator in the term until they enclose just the single-qubit Pauli operator affecting qubit $j$.\\
As qubit $j$ is a $\sigma_k$-exclusion qubit, the Pauli acting on qubit $j$ in any term in $\mathcal{H}$ will not be $\sigma_k$. Hence, using Lemma \ref{negationTheorem}, the effect of the conjugation is to add an overall factor of $-1$ to the term.
As this happens for every term, an overall phase of $-1$ is added to $\mathcal{H}$.
\end{proof}
The Hamiltonian can be conjugated by as many Pauli operators, on different qubits as required, and each will still produce the factor of $-1$.

\subsection{A Graph Theory Aside}
To invert a time evolution where there are many XY-interaction terms we will need to start with a slight diversion into graph theory.
\begin{definition}{A Graph}\( \\ \)
    A graph, $\mathcal{G} = (V, E)$, is a set of two sets: $V$ and $E$.
    Elements of $V$ are referred to as vertices and each element in $E$ (referred to as an edge) is a double of elements in $V$.
\end{definition}

\begin{definition}{k-Colourability}\( \\ \)
\label{kColDef}
    A graph, $\mathcal{G}$, is k-colourable if and only if there exists a mapping, $\chi$: $V \longrightarrow$ $\{1, 2, ..., k\}$ such that:
    \begin{align}
        &\forall (v_1, v_2) \in E,\\
        &\chi(v_1) \neq  \chi(v_2)
    \end{align}
\end{definition}
\begin{definition}{Chromatic Subsets, $\Tilde{\mathcal{V}}_j$}\( \\ \)
\label{cromatSetDef}
Define the chromatic subsets for a graph $\mathcal{G}$ by:
\begin{align}
    \Tilde{V}_j &= \big \{ v \in V \vert \chi(v) = j \big \}
\end{align}
Where the underlying graph, $\mathcal{G}$, has been k-coloured for some integer $k \geq j$, and $\chi$ is the mapping as in Definition \ref{kColDef}.
\end{definition}

\begin{lemma}\( \\ \)
    \label{OneInEach}
    If a graph, $\mathcal{G} = (V, E)$, is two-coloured, every edge in $\mathcal{G}$ has exactly one endpoint in each chromatic subset. 
\end{lemma}
\begin{proof}
    Follows from the definition of a chromatic set of a two colouring.
\end{proof}
\subsection{Inversion of Relevant Hamiltonians}
The notion of graphs is useful for our purposes due to a mapping between graphs and Hamiltonians that involve two body interactions only.
This mapping is not bijective: it is invariant under varying non-zero interaction strengths in the Hamiltonian, but this gives the mapping general applicability.
The mapping is defined by identifying a vertex with each qubit the Hamiltonian affects and if there is a two-body interaction between any two qubits (in the Hamiltonian) adding an edge between the corresponding vertices.

\begin{proof}[Proof of Theorem~\ref{fullInversionXYTheorem}]
    As shown in Lemma \ref{OneInEach}, every edge in the graph generated by the Hamiltonian is adjacent to exactly one qubit in the chosen chromatic set, $\Tilde{\mathcal{V}}_1$. Considering a general XY interaction between qubits $j$ and $k$:
    \begin{align}
        X_j X_k + Y_j Y_k
    \end{align}
    Assume WLOG, $j \in \Tilde{\mathcal{V}}_1$. Then, via Lemma \ref{OneInEach}, $k \not \in \Tilde{\mathcal{V}}_1$.\\
    So the only effect of conjugating every qubit in $\Tilde{\mathcal{V}}_1$ with a Pauli $Z$ operator is conjugating qubit $j$ with a Pauli $Z$ gate:
    \begin{align}
        Z_j\bigg( X_j X_k + Y_j Y_k \bigg) Z_j^{\dag}
        &=
        Z_j X_j Z_j^{\dag} X_k + Z_j Y_j Z_j^{\dag} Y_k
    \end{align}
    Then using Lemma \ref{negationTheorem}:
    \begin{align}
        Z_j\bigg( X_j X_k + Y_j Y_k \bigg) Z_j^{\dag}
        &=
        - \big( X_j X_k + Y_j Y_k \big)
    \end{align}
    If $\langle j, k \rangle$ denotes the set of all pairs $(j,k)$ such that vertices indexed by $j$ and $k$ are neighbours on the graph, the Hamiltonian, $\mathcal{H}$, can be expressed as:
    \begin{align}
        \mathcal{H} 
        &=
        \sum_{\langle j, k \rangle} \bigg( X_jX_k + Y_jY_k \bigg)
    \end{align}
    If we then apply the conjugation on the chromatic set, the previous consideration of single interactions implies:
    \begin{align}
        \prod_{q \in \Tilde{\mathcal{V}}_1} \bigg( Z_q \bigg) \mathcal{H}
        \prod_{q \in \Tilde{\mathcal{V}}_1} \bigg( Z_q\bigg)
        &=
        - \mathcal{H}
    \end{align}
\end{proof}
One such example of the above kind of Hamiltonian is the XY-interaction on a square lattice.

\subsection{Inverting Accreditable Hamiltonians}
\label{InvertXYSubsection}
We first define what we mean by inversion of a time evolution.
\begin{definition}\( \\ \)
    Any operations that map a time evolution, $e^{-i \mathcal{H} t}$, to its reversed time evolution, $e^{i \mathcal{H} t}$, is said to invert the time evolution.
\end{definition}

Before proceeding, we need to prove Lemma \ref{unitExp} from the main text:
\begin{proof}[Proof of Lemma \ref{unitExp}]
    Starting from the Taylor expansion of the time evolution operator
    \be
        e^{- \kappa \mathcal{H} t}
        =
        \label{expTaylorSeries}
        \sum^{\infty}_{j = 1} \bigg( \dfrac{ \big( - \kappa \mathcal{H} t \big)^j}{j!}  \bigg),
    \ee
    the result follows from the conjugation of Eqn. \ref{expTaylorSeries} by the unitary, $\mathcal{U}$,
    \be
        \mathcal{U} e^{- \kappa \mathcal{H} t} \mathcal{U}^{\dag}
        =
        \mathcal{U} \sum^{\infty}_{j = 1} \bigg( \dfrac{ \big( - \kappa \mathcal{H} t \big)^j}{j!}  \bigg) \mathcal{U}^{\dag}
        =
        \sum^{\infty}_{j = 1} \bigg( \dfrac{ \big( - \kappa \mathcal{U} \mathcal{H} \mathcal{U}^{\dag} t \big)^j}{j!} \bigg).
    \ee
\end{proof}

This is then used in the proof of Theorem \ref{fullTimeInversion}

\begin{theorem}\( \\ \)
    \label{fullTimeInversion}
    If $\mathcal{H}$ is a Hamiltonian consisting of XY-interactions on a square lattice, conjugating the time evolution according to $\mathcal{H}$ (for any duration) with Pauli $Z$ gates on any chromatic set of the square lattice inverts the time evolution.
\end{theorem}
\begin{proof}
    As applying a Pauli $Z$ to every qubit in a chromatic set, $\Tilde{\mathcal{V}}$, is a unitary, denoting $\mathcal{C} = \prod_{j \in \Tilde{\mathcal{V}}} \big( Z_j \big)$, Lemma \ref{unitExp} implies,
    for any time duration, $t \in \mathbb{R}$:
    \begin{align}
        \mathcal{C} e^{-i \mathcal{H} t} \mathcal{C}^{\dag}
        &=
        e^{-i \mathcal{C} \mathcal{H}\mathcal{C}^{\dag} t}
    \end{align}
    Then Theorem \ref{fullInversionXYTheorem} implies:
    \begin{align}
        \mathcal{C} e^{-i \mathcal{H} t} \mathcal{C}^{\dag}
        &=
        e^{i \mathcal{H} t}
    \end{align}
    This is the inverted time evolution.
\end{proof}

\begin{definition}{Checkerboard Set on a Square Lattice}\( \\ \)
\label{CheckDef}
For a $N_l \times N_l$ square lattice, index the qubits from the top left, starting from 1, working along each row, then moving to the next one down, always starting from the left of a row.\\
If $N_l$ is odd, its checkerboard set, $\mathcal{S}(N_l)$, is:
\begin{align}
    \big \{ 1 \leq j \leq N_l^2 \big \vert\textit{ } j \textit{ is odd } \big \}
\end{align}
If $N_l$ is even, find $\mathcal{S}(N_l+1)$ and remove the rightmost column and bottom row.
\end{definition}
\begin{lemma}\( \\ \)
    \label{SNis2Colour}
    $\forall N_l \in \mathbb{Z}^+$, $\mathcal{S}(N_l)$ is a chromatic set for a 2-colouring.
\end{lemma}
\begin{proof}
    Removing vertices from a graph (and their adjacent edges), leaves chromatic sets invariant, if the removed vertices (from the graph) are also removed from the chromatic sets.\\
    Hence we can assume $N_l$ is odd and only consider 2-colouring the degree $4$ vertices (which are identically connected) of a $(N_l+2) \times (N_l+2)$ square lattice.\\
    If a vertex has index $j$, then, assuming it is a degree $4$ vertex, it is adjacent to the vertices with indices: $j - 1$, $j + 1$, $j - N_l$, $j + N_l$.\\
    As both $1$ and $N_l$ are odd, $j$ has a different parity to $j - 1$, $j + 1$, $j - N_l$, and $j + N_l$.\\
    Therefore, $j$ is in $\mathcal{S}(N_l)$ if and only if all vertices it is adjacent to are not in $\mathcal{S}(N_l)$. This is true for any vertex so $\mathcal{S}(N_l)$ is a chromatic set for a 2-colouring of a $N_l \times N_l$ square lattice.
\end{proof}
\subsection{Proof of Theorem \ref{timeReversalExistTheorem}}

\begin{theorem}
\label{timeReversalExistTheorem}
    For any accreditable Hamiltonian, a time inversion circuit exists and can be built by applying a Pauli Z to every qubit in the  checkerboard set corresponding to the Hamiltonian's underlying square lattice.
\end{theorem}

\begin{proof}
As in Definition \ref{accreditableHamDef}, accreditable Hamiltonians consist of the XY-interaction on a 2D square lattice.\\
Then Lemma \ref{SNis2Colour} implies the Hamiltonian's underlying checkerboard set is a chromatic set of the underlying graph.\\
Theorem \ref{fullTimeInversion} implies that this is sufficient for applying a Pauli Z operator to every qubit in the Hamiltonian's underlying checkerboard set to constitute a time inversion circuit.
\end{proof}

\section{Time Inversion Examples}
\label{ExamplesAppendix}

\subsection{Simple 2 qubit System}
One simple example of a  Hamiltonian that can be inverted is:
\begin{align}
    \label{SingleXYInteractHamil}
    \mathcal{H} 
    &=
    X_1X_2 + Y_1 Y_2
\end{align}
In this case, Lemma \ref{HamiltonianPhase} tells us that conjugation with either $Z_1$ or $Z_2$ adds a factor of $-1$:
\begin{align}
    Z_1\mathcal{H} Z_1
    &=
    -Z_1 \cdot Z_1 X_1X_2 - Z_1 \cdot Z_1Y_1Y_2\\
    &= 
    \label{refTimeReverse}
    - \mathcal{H}
\end{align}
The time evolution according to Hamiltonian, $\mathcal{H}$ (as above), for time, $t$, can then be conjugated by $Z_1$. Hence, Lemma \ref{unitExp} and Eqn.~\ref{refTimeReverse} imply:
\begin{align}
        Z_1 e^{-i \mathcal{H} t} Z_1^{\dag}
        &=
        e^{i \mathcal{H} t}
    \end{align}
    We see that the conjugation of the time evolution by $Z_1$ has reversed the time evolution's direction.\\

    Clearly this example is not on a square lattice and is not directly used for our accreditation protocol. It does however provide a simple illustrative example and also shows these techniques to have a wider applicability.
\subsection{XY Model on a 3 by 3 Square Lattice}
\subsubsection{Derivation of a Time Inversion Circuit}
We can in fact invert the full XY-model, which differs from the XY-interaction by adding onsite terms (Pauli $Z$ gates with real coefficents on every qubit).\\
This is a wider application of time inversion circuits and although similar to what is required in the accreditation protocol, this is the XY-model. The accreditable Hamiltonians consist entirely of the XY-interaction i.e. have no onsite terms. \\

For our example, we choose the XY-model (that is, the XY interaction with on-site terms consisting of a Pauli Z on each site) on a 3x3 square lattice, with nearest neighbour interactions.\\
The sites in the lattice can be indexed as in Fig. \ref{LatticePic}. The edges of the graph then denote the interactions, which helps clarify the model.
\begin{figure}
    \centering
    \includegraphics{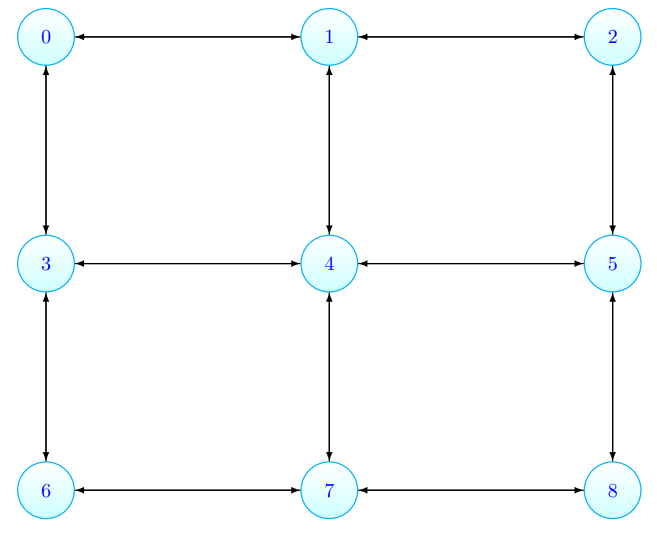}
    \caption{A $3 \times 3$ square lattice with an indexing on the vertices}
    \label{LatticePic}
\end{figure}
Using this labelling, we can define the Hamiltonian of the system:
\begin{align}
    \mathcal{H} 
    &= 
    \sum_{\langle j,k \rangle } \bigg( J_{i,j} \bigg[ X_{j} X_{k} + Y_{j} Y_{k} \bigg] \bigg) 
    + U \sum_{j = 0}^{8} \bigg( Z_j \bigg)
\end{align}
Where,\\
$\bullet$ $\langle j,k\rangle$ denotes any $j, k$ such that the nodes labelled $j$ and $k$ are neighbours.\\
$\bullet$ $\forall i,j \in \langle j,k\rangle$, $J_{i,j} \in \mathbb{R}$.\\
$\bullet$ $ U \in \mathbb{R}$.\\
As they end up not affecting how the system gets inverted, we are free to entirely neglect $J_{i,j}$ and $U$ i.e. set each to $1$. Hence, the Hamiltonian can be considered to be:
\begin{align}
    \mathcal{H} 
    &= 
    \sum_{\langle j,k \rangle } \bigg( X_{j} X_{k} + Y_{j} Y_{k} \bigg) 
    + \sum_{j = 0}^{8} \bigg( Z_j \bigg)
\end{align}
Where $\langle j,k\rangle$ denotes any indices $j, k$ such that the nodes labelled $j$ and $k$ are neighbours.\\

If we can find two circuits: one which acts as a time inversion circuit for the first summation in the Hamiltonian (i.e. the XY interaction terms) and commutes with the second summation (i.e. the on-site terms), and the other that commutes with the first term but acts as a time inversion circuit for the second summation in the Hamiltonian; then the composition of the two circuits acts as a time inversion circuit for the entire Hamiltonian i.e. each part of the Hamiltonian is inverted by a different circuit.\\

\noindent \underline{Inverting the XY Interactions}\\
Using Theorem \ref{fullInversionXYTheorem}, a time inversion circuit for the XY interactions on a square lattice can be formed by 2-colouring the lattice and applying a Pauli $Z$ gate on each site in one chromatic set.\\
This will not affect the onsite terms (i.e. the other summation in the Hamiltonian) as this circuit is formed entirely of Pauli $Z$ gates and so are the onsite terms, so they commute.\\
 To 2-colour the graph in Fig. \ref{LatticePic} (which we henceforth refer to as $\mathcal{G}$), as is required to invert the Hamiltonian acting on it, split the graph into the below defined subgraph, $\mathcal{S}$, and it is compliment (on the graph).
\begin{align}
    \mathcal{S} &= \{ x \in \mathcal{V}(\mathcal{G}) \vert  \mathcal{V}(\mathcal{G}) \textit{ is odd}  \}
\end{align}
By inspection of the graph in Fig. \ref{LatticePic}, we can see $\mathcal{S}$ 2-colours it. It is in fact a checkerboard set as in Definition \ref{CheckDef}.\\

\noindent \underline{Inverting the Onsite $Z$ terms}\\
If a time inversion circuit for the onsite terms is constructed of a Pauli $X$ gate on every site in the Hamiltonian, then each term in the second summation in the Hamiltonian gains a factor on minus one due to Lemma \ref{negationTheorem} i.e. the onsite terms are inverted.\\
Additionally, it can easily be checked that:
\begin{align}
    \bigg[ X_1X_2, X_1X_2 + Y_1Y_2 \bigg]
    &=
    0
\end{align}
So conjugation by Pauli $X$ gates has no effect on the interaction terms.\\

\noindent \underline{Inverting the Entire Hamiltonian}\\
With a time inversion circuit for each summation in the Hamiltonian, each of which has no effect on the summation it does not invert, we can construct a time inversion circuit for the entire Hamiltonian composing the two circuits i.e.
\begin{align}
    \mathcal{C} 
    &=
    \prod_{k \in \mathcal{S}} \bigg( Z_k \bigg) \prod_{k = 0}^8 \bigg( X_k \bigg)
\end{align}
This construction is the most useful from the perspective of our analysis. It is not, however, the most efficient implementation.\\
If we simplify via Pauli matrix multiplication:
\begin{align}
    \mathcal{C} 
    &=
    \prod_{k \in \mathcal{S}} \bigg( iY_k \bigg) \prod_{k \not \in \mathcal{S}} \bigg( X_k \bigg)
\end{align}

\subsubsection{Demonstration of Time Evolution Inversion}
We now show that the circuit achieves inversions as intended.
\begin{align}
\label{InitialEffect}
    \mathcal{C}^{\dagger} e^{-i \mathcal{H} t} \mathcal{C}
    &=
    e^{-i  \mathcal{C}^{\dagger} \mathcal{H} \mathcal{C} t}
\end{align}
Then, examining just $\mathcal{C}^{\dagger} \mathcal{H} \mathcal{C}$:
\begin{align}
    \mathcal{C}^{\dagger} \mathcal{H} \mathcal{C}
    &=
     \prod_{k = 0}^8 \bigg( X_k \bigg) \prod_{k \in \mathcal{S}} \bigg( Z_k \bigg) \bigg[  \sum_{\langle j,k \rangle } \bigg( X_{j} X_{k} + Y_{j} Y_{k} \bigg) 
    \bigg] \prod_{k \in \mathcal{S}} \bigg( Z_k \bigg) \prod_{k = 0}^8 \bigg( X_k \bigg)
    +
    \sum_{j = 0}^{8} \bigg( X_jZ_jX_j \bigg)
\end{align}
Due to the 2-colouring of the interaction graph, one of every node involved in each interaction is in the sub-graph generated by the 2-colouring. For ease, assume that in each case the node in the sub-graph, $\mathcal{S}$, is the one labelled $k$.
\begin{align}
    \mathcal{C}^{\dagger} \mathcal{H} \mathcal{C}
    &=
    \sum_{\langle j,k \rangle } \bigg( - X_j X_{j} X_j X_k X_{k} X_k - X_j Y_{j} X_j X_k Y_{k} X_k \bigg)
    +
    \sum_{j = 0}^{8} \bigg( - Z_j\bigg)\\
    &= 
    \label{finalInversionExample}
    - \mathcal{H}
\end{align}
Using Eqn.~\ref{finalInversionExample} in Eqn.~\ref{InitialEffect}:
\begin{align}
    \mathcal{C}^{\dagger} e^{-i \mathcal{H} t} \mathcal{C}
    &=
    e^{i \mathcal{H} t}
\end{align}
So the overall effect is an inversion of the time evolution.

\end{document}